\documentclass[aip,amsmath,amssymb,preprint]{revtex4-1}

\usepackage{graphicx}
\usepackage{bm}
\usepackage{amsthm}
\usepackage{fullpage}
\usepackage{nicefrac}
\usepackage{tikz}
\usetikzlibrary{decorations.markings}
\usepackage{float}
\usepackage[utf8]{inputenc}
\usepackage[T1]{fontenc}
\usepackage{mathptmx}
\usepackage{etoolbox}

\newcommand{\derh}[2]{\frac{\partial H}{\partial #1_{#2}}}
\newcommand{\derf}[3]{\frac{\partial f_{#1}}{\partial #2_{#3}}}
\newcommand{\derderf}[5]{\frac{\partial^2 f_{#1}}{\partial #2_{#3} \partial #4_{#5}}}

\newcommand{\derderh}[4]{\frac{\partial^2 H}{\partial #1_{#2} \partial #3_{#4}}}
\newcommand{\derq}[2]{\frac{\partial \vec{q}_{#1}}{\partial \vec{b}_{#2}}}
\newcommand{\derp}[2]{\frac{\partial \vec{p}_{#1}}{\partial \vec{b}_{#2}}}
\newcommand{\p}{\vec{p}}
\newcommand{\q}{\vec{q}}
\newtheorem{theorem}{Theorem}[section]
\newtheorem{corollary}{Corollary}[theorem]
\newtheorem{lemma}[theorem]{Lemma}
\newtheorem{definition}{Definition}[section]

\makeatletter
\def\@email#1#2{%
 \endgroup
 \patchcmd{\titleblock@produce}
  {\frontmatter@RRAPformat}
  {\frontmatter@RRAPformat{\produce@RRAP{*#1\href{mailto:#2}{#2}}}\frontmatter@RRAPformat}
  {}{}
}%
\makeatother

\begin{document}
\title{Generalized Gelfand-Yaglom Formula for a Discretized Quantum Mechanic System}
\author{Meredith Shea}
\altaffiliation{mshea@berkeley.edu}
\affiliation{University of California Berkeley, Department of Mathematics, Berkeley, United States}

\date{\today}

\begin{abstract}
    The Gelfand-Yaglom formula relates the regularized determinant of a differential operator to the solution of an initial value problem. Here we develop a generalized Gelfand-Yaglom formula for a Hamiltonian system with Lagrangian boundary conditions in the discrete and continuous settings. Later we analyze the convergence of the discretized Hamilton-Jacobi operator and propose a lattice regularization for the determinant.
\end{abstract}

\maketitle 

\section{Introduction} \label{section1}
\subsection{The Gelfand-Yaglom Formula} \label{section1.1}
In Gelfand and Yaglom\cite{GY}, they study the evaluation of certain integrals with respect to the Wiener measure. In their paper, they found that the solution to certain integrals of exponentials can be expressed in terms of a solution to a Sturm-Liousville problem. Later this formula was interpreted as a relation between the regularized-determinant of an elliptic operator and the solution to an initial value problem. Let us start with an overview of this formula. Consider a one dimensional quantum mechanic system with potential $V(q)$. The action functional on the space of paths is,
\begin{equation}
    S\left[\gamma\right] = \int_0^{T} \left(\frac{m}{2}\dot{q}(t)^2 - V\left(q(t)\right) \right) dt \label{eq:actionL}
\end{equation}
where $\gamma: \, [0,T] \to \mathbb{R}$, $t \mapsto q(t)$ is a path. We have adopted the usual notation $\dot{q} = dq/dt$. Let $\gamma_c$ be a critical point of this function, i.e. $\gamma_c(t) = q_c(t)$ solves the differential equation,
\begin{equation}
    m\ddot{q}_c(t) = -V'\left(q_c(t)\right) \label{eq:ELeq}
\end{equation}
with boundary conditions,
$$q_c(0) = q \quad \quad \dot{q}_c(0) = \frac{p}{m}$$
where $q$ and $p$ are parameters describing the initial position and momentum. Let $A$ denote the differential operator which appears in the second variation of the action functional,
\begin{equation*}
    \delta^2S\left[\gamma_c\right] = \int_0^T \delta q \, A \, \delta q \; dt
\end{equation*}
For the case of \eqref{eq:actionL} the operator $A$ is explicitly defined by,
\begin{equation}
    A=-\frac{d^2}{dt^2} - \frac{1}{m} V''\left(q_c\right)\label{eq:opA}
\end{equation}
and is equipped with Dirichlet boundary conditions. The Gelfand-Yaglom formula states, 
\begin{equation}
    \frac{\partial q_c(T)}{\partial p} = \frac{1}{2m}\det{}_\zeta(A) \label{eq:GY1}
\end{equation}
where $p$ is the parameter defined above. \\

\noindent To make sense of the above formula we must first define the $\zeta$-regularized determinant of an operator. Let $L$ be a differential operator with a discrete spectrum that is bounded from below. Remove zero eigenvalues and enumerate the spectrum, $\lambda_1 \leq \lambda_2 \leq \cdots \leq \lambda_n \leq \cdots$. Assuming the following series converges for sufficiently large $\Re(s)$, we define the $\zeta$-function of the operator $L$ to be, 
\begin{equation*}
    \zeta_L(s) = \sum_i \frac{1}{\lambda_i^s}
\end{equation*}
The $\zeta$-regularized determinant is defined as,
\begin{equation*}
    \det{}_\zeta (L) = e^{-\zeta_L'(0)}
\end{equation*}
where we must analytically continue the derivative of the $\zeta$-function of the operator to the point $s = 0$. Note that this is possible by Seeley's theorem, which states that the zeta-function of an elliptic operator extends to a meromorphic function in the complex plane and the origin is always a regular point. In the case of second order differential operators, see Takhtajan\cite{QMmath} and Kirsten\cite{FuncDet} for examples of computing $\zeta$-regularized determinants. 

\subsection{The Gelfand-Yaglom Formula in the Hamiltonian Formalism} \label{section1.2}
Now let us rewrite the system expressed in \eqref{eq:actionL} in terms of the Hamiltonian formalism. We let 
\begin{equation*}
    \tilde{\gamma}: \; [0,T] \to \mathbb{R}^2 = \left\{\left(p(t),q(t)\right)\right\}
\end{equation*} 
represent a path in the phase space. That is, we express a path in terms of its coordinates on the cotangent bundle of $\mathbb{R}$. The Hamilton-Jacobi action is,
\begin{equation}
    \tilde{S}[\tilde{\gamma}] = \int_0^T \Big(p(t) \dot{q}(t) - H\big(p(t),q(t)\big)\Big) \; dt \label{eq:actionH}
\end{equation}
where $H\big(p(t),q(t)\big)$ is the Hamiltonian corresponding to the system in equation \eqref{eq:actionL}. In other words, it is the Legendre transformation of the Lagrangian. Explicitly it is
\begin{equation}
    H\big(p(t),q(t)\big) = \frac{p(t)^2}{2m} + V\big(q(t)\big)\label{eq:coreH}
\end{equation}
Critical values of the action in \eqref{eq:actionH} are solutions to Hamiton's equations, denoted $\tilde{\gamma}_c(t) = \big(p_c(t), q_c(t)\big)$. 
\begin{equation*}
    \dot{q}_c(t) = \frac{1}{m} p_c(t), \quad \dot{p}_c(t) = -V' \big(q_c(t)\big)
\end{equation*}
Let us denote the critical values of $\tilde{S}$ with the notation, 
\begin{equation*}
    \tilde{S}_{\tilde{\gamma}_c}(q,q') = \tilde{S}[\tilde{\gamma}_c]
\end{equation*}
where $q = q_c(0)$ and $q' = q_c(T)$ define the starting and ending positions of the path $\tilde{\gamma}_c$. Clearly $\tilde{S}[\tilde{\gamma}_c] = S[\gamma_c]$. A quick computation yields,
\begin{equation*}
    \frac{\partial q_c(T)}{\partial p} = \left( \frac{\partial \tilde{S}_{\tilde{\gamma}_c}(q,q')}{\partial q \partial q'} \right)^{-1}
\end{equation*}
Inserting the above into \eqref{eq:GY1} gives a GY formula in terms of the phase space formalism,
\begin{equation}
    \left( \frac{\partial \tilde{S}_{\tilde{\gamma}_c}(q,q')}{\partial q \partial q'} \right)^{-1} = \frac{1}{2m} \det{}_\zeta \, A \label{eq:GY2}
\end{equation}
Where $A$ is again the operator given by equation \eqref{eq:opA}.

\subsection{An Action Functional with Lagrangian Boundary Conditions} \label{section1.3}
We now amend the action functional in \eqref{eq:actionH} by defining functions $f_1, \; f_2: \; \mathbb{R}^2 \to \mathbb{R}$. Explicitly, $f_1$ is a function of the initial position $q = q(0)$ and a parameter $b_1$, while $f_2$ is a function of the final position $q' = q(T)$ and a parameter $b_2$. These functions define Lagrangian boundary conditions on the phase space. The generalized action functional can be written as, 
\begin{equation}
    \Tilde{S}[\tilde{\gamma}] = \int_0^T \Big(p(t) \dot{q}(t) - H\big(p(t),q(t)\big)\Big) \; dt + f_1 (q, b_1) - f_2 ( q', b_2 ) \label{eq:actionH2}
\end{equation}
We assume $\mathbb{R}^2$ has the standard symplectic structure with coordinates $(p,q)$ and symplectic form $\omega = dp \wedge dq$. For now we will suppose $H\big(p(t),q(t)\big)$ is an arbitrary Hamiltonian that is at least twice differentiable in both variables. The critical points of the above generalized action functional are solutions to the boundary problem,
\begin{equation*}
    \dot{p}(t) = - \frac{\partial H}{\partial q(t)}\big(p(t),q(t)\big) \quad \quad \dot{q}(t) = \frac{\partial H}{\partial p(t)}\big(p(t),q(t)\big)
\end{equation*}
where
\begin{equation}
    p(0) = \frac{\partial f_1}{\partial q} \quad \quad p(T) = \frac{\partial f_2}{\partial q'} \label{eq:LBC}
\end{equation}
Thus critical points are flow lines of the Hamiltonian vector field generated by $H$, connecting the following two Lagrangian submanifolds 
\begin{equation*}
    \mathcal{L}_1 = \left\{ (p,q) \; | \; p = \frac{\partial f_1}{\partial q}(q,b_1) \right\}
\end{equation*} 
\begin{equation*}
    \mathcal{L}_2 = \left\{ (p,q') \; | \; p = \frac{\partial f_2}{\partial q'}(q',b_2) \right\}
\end{equation*} 
in time $T$. The second variation of the action in (8) near the classical trajectory defines a first order differential operator $\tilde{A}$,
\begin{equation*}
    \delta^2 \Tilde{S}[\gamma_c] = \int_0^T (\delta p, \; \delta q) \Tilde{A}\begin{pmatrix} \delta p \\ \delta q \end{pmatrix} \; dt
\end{equation*}
where $\tilde{A}$ is defined explicitly as, 
\begin{equation}
    \Tilde{A} = 
    \begin{pmatrix}
    -\frac{\partial^2 H}{\partial p^2}(p_c,q_c) & \frac{d}{dt} - \frac{\partial^2 H}{\partial q  \partial p}(p_c,q_c) \\
    -\frac{d}{dt} - \frac{\partial^2 H}{\partial q  \partial p}(p_c,q_c) & -\frac{\partial^2 H}{\partial q^2}(p_c,q_c)
    \end{pmatrix} \label{eq:optildeA}
\end{equation}
with boundary conditions, 
\begin{equation}
    x_2(0) = \frac{\partial^2 f_1}{\partial q^2}(q,b_1)  x_1(0) \quad \quad x_2(T) = \frac{\partial^2 f_2}{\partial q'^2}(q',b_2) x_1(T) \label{eq:LBCvec}
\end{equation}
where $\tilde{A}$ acts on the transpose of the vector $\left(x_1(t)\\x_2(t)\right)$. The boundary conditions above translate to the mixed boundary conditions when considering the second order differential operator $A$. Explicitly, if $A$ acts on the function $y(t)$ we can express the mixed boundary conditions as, 
\begin{equation}
    y'(0) = \frac{1}{m}\frac{\partial^2 f_1}{\partial q^2}(q,b_1) y(0) \quad \quad y'(T) = \frac{1}{m} \frac{\partial^2 f_2}{\partial q'^2}(q',b_2) y(T) \label{eq:LBCmixed}
\end{equation}

\noindent We desire a GY formula that is analogous to equation (7) which uses the operator $\tilde{A}$ and the action functional from equation (8), however the $\zeta$-regularized determinant of a first order operator depends on the choice of the spectral cut in the plane. 
\\

In this paper we give an alternative regularization of the operator $\tilde{A}$, which we refer to as the lattice-regularization. We compare this proposed regularization to the $\zeta$-regularization with the goal of showing that they are agreeable. Moreover, computation of the  lattice-regularization is simply a problem in matrix determinants. \\

The structure of the paper is as follows. In section \ref{section2}, we develop a discrete model of the system described above. From this discrete model we develop a discrete Gelfand Yaglom formula. At the end of section \ref{section2} we also develop a discrete model of the operator $A$ from \eqref{eq:opA} and compare it to the discrete operator form of $\tilde{A}$ from \eqref{eq:optildeA}. Section \ref{section3} is devoted to the convergence of these discrete operators in the continuum limit. After proving convergence we are able to define a lattice regularization for the determinants of $A$ and $\tilde{A}$. Lastly, in section \ref{section4} we prove a similar generalized Gelfand Yaglom formula for the operator $A$ using the $\zeta$-regularized determinant. From this, we are finally able to compare the $\zeta$-regularization to the lattice regularization for the determinant of the operator $A$.

\section{A Discretized Generalized Gelfand-Yaglom Formula} \label{section2}
\subsection{Discretized Quantum Mechanics System} \label{section2.1}
In this section, we will develop a discretized version of the usual quantum mechanics system. In this discrete setting, all determinants will be finite. This allows us to compute the following with ease: a generalized Gelfand-Yaglom formula in the Hamiltonian formalism, and a relationship between the determinants of the discretized versions of the operators $A$ and $\tilde{A}$. Later, in section \ref{section3}, we will consider how these results behave in the continuum limit, thus defining an alternative regularization for the determinants of $A$ and $\tilde{A}$.\\
\begin{figure}
\centering
\includegraphics{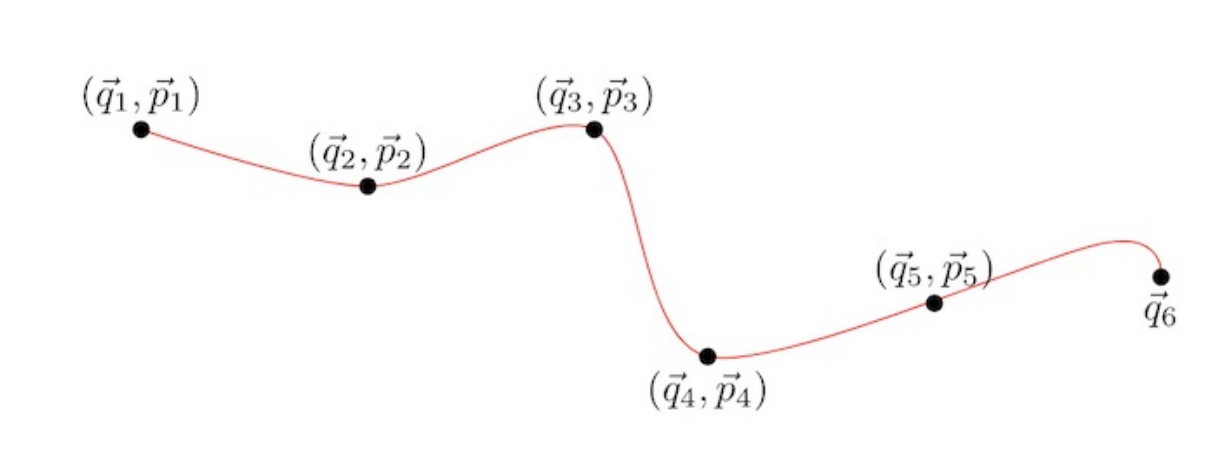}
\caption{\label{fig:fig1} A discretization of a path into $N = 6$ position vectors and $N-1 = 5$ momentum vectors.}
\end{figure}
First we discretize any given path, $\tilde{\gamma}\big(\q(t),\p(t)\big): \; [0,T] \to \mathbb{R}^{2n}$, into $N$ position and $N-1$ momentum vectors as shown in figure \ref{fig:fig1}, where $\q_i = \q\big((i-1)\cdot\epsilon\big)$ and $\p_i = \p\big((i-1)\cdot\epsilon\big)$ and $\epsilon = \frac{T}{N}$. From the above discretization and the action given in equation \eqref{eq:actionH2}, we propose the following discrete action functional
\begin{equation}
    \tilde{S}_d[\tilde{\gamma}_d] = \sum_{i=1}^{N-1} \p_i(\q_{i+1}-\q_i) - \sum_{i=1}^{N-1} H(\p_i,\q_i) - f_2(\q_N,\vec{b}_2) + f_1(\q_1,\vec{b}_1) \label{eq:actionH2D}
\end{equation}
where $f_1$ and $f_2$ are the same functions appearing in equation \eqref{eq:actionH2}. We will only consider discrete Hamiltonians that arise from twice differentiable continuous Hamiltonians. Note that $\q_1 = \q(0) = q$ and $\q_N = \q(T) = q'$ are exactly what appears in the continuous statement of the action. From the above we derive a discrete version of Hamilton's equations,
\begin{align}
    &\q_{i+1}-\q_i-\derh{\p}{i}(\p_i,\q_i) = 0 \quad i = 1, \dots, N-1 \label{eq:hamD1} \\
    &\p_i-\p_{i-1}+\derh{\q}{i}(\p_i,\q_i) =0 \quad i = 2, \dots, N-1 \label{eq:hamD2} 
\end{align}
and the boundary conditions,
\begin{align}
    &\derf{1}{\q}{1}=\p_1+\derh{\q}{1}(\p_1,\q_1) \label{eq:DBC1} \\
    &\derf{2}{\q}{N} = \p_{N-1} \label{eq:DBC2}
\end{align}
which agree with the conditions from \eqref{eq:LBC} in the continuum limit. The discretized path $\tilde{\gamma}_{d,c} = \{\p_1, \dots, \p_{N-1}, \q_1, \dots, \q_N\}$ that satisfies equations \eqref{eq:hamD1}-\eqref{eq:DBC2} will be known as the (discrete) critical point or classical path.\\

\noindent When we take the second variation of the discretized action functional we obtain a matrix operator which acts on the vector $\delta \tilde{\gamma}_{d,c}$ in the following manner, 
\begin{equation*}
    \delta^2 S_d[\tilde{\gamma}_{d,c}] = \delta \tilde{\gamma}_{d,c} \tilde{A}_N (\delta \tilde{\gamma}_{d,c})^T
\end{equation*}
The matrix $\tilde{A}_N$ is the discrete analog of the Hamilton-Jacobi operator $\tilde{A}$ with $N$ time intervals. To define $\tilde{A}_N$ explicitly, we should first note that it has a block form, 
\begin{equation}
    \tilde{A}_N = 
    \begin{pmatrix}
    D_1 & D_2 \\
    D_3 & D_4
    \end{pmatrix} \label{eq:blocktildeAN}
\end{equation}
In the one-dimensional case the block above can be written explicitly as, 
\begin{equation*}
    (D_1)_{ij} = 
    \begin{cases}
    -\derderh{p}{i}{p}{i} & \text{ if } i = j \\
    0 & \text{ if } i \neq j
    \end{cases}
\end{equation*}
\begin{equation*}
    (D_2)_{ij} = (D_3)_{ji} = 
    \begin{cases}
    -1- \derderh{p}{i}{q}{i} & \text{ if } i = j \\
    1 & \text{ if } i+1 = j \\
    0 & \text{ otherwise }
    \end{cases}
\end{equation*}
\begin{equation*}
    (D_4)_{ij} = 
    \begin{cases}
    \derderf{1}{q}{1}{q}{1} - \derderh{q}{1}{q}{1} & \text{ if } i = j = 1 \\
    -\derderh{q}{i}{q}{i} & \text{ if } 2 \leq i = j \leq N-1 \\
    -\derderf{2}{q}{N}{q}{N} & \text{ if } i = j = N \\
    0 & \text{ if } i \neq j
    \end{cases}
\end{equation*}
where all derivatives are taken at the critical point. These matrices are immediately generalized to the $n$-dimensional case, where partial derivatives become $n \times n$ matrices of partial derivative and any constant is multiplied by the $n \times n$ identity matrix. 

\subsection{Generalized Gelfand-Yaglom Formula} \label{section2.2}
We will restrict our Hamiltonians to those that satisfy, 
\begin{equation*}
    \det \left(\mathbb{I} + \derderh{\vec{p}}{i}{\vec{q}}{i} \right) \neq 0
\end{equation*}
and
\begin{equation*}
     \det \left( \derderh{\p}{i}{\p}{i} \right) \neq 0
\end{equation*}
for all $i = 1, \dots, N-1$. From the above set up we derive a generalized GY formula. Note all determinant below are determinants of finite matrices. 
\begin{theorem}
The discrete action functional defined by \eqref{eq:actionH2D} satisfies the generalized Gelfand-Yaglom formula
\begin{equation}
    \det\left(\frac{\partial^2 \tilde{S}_{d,\tilde{\gamma}_c}(b_1,b_2)}{\partial \vec{b}_1 \partial \vec{b}_2}\right) = \prod_{i=1}^{N-1} \det\left(-\derderh{\p}{i}{\q}{i}-\mathbb{I}\right) \frac{\det\left( \derderf{1}{\q}{1}{\vec{b}}{1} \right)\det\left(\derderf{2}{\q}{N}{\vec{b}}{2}\right)}{\det\tilde{A}_N}\label{eq:DGYgen}
\end{equation}
where $\tilde{S}_d[\tilde{\gamma}_c] = \tilde{S}_{d,\tilde{\gamma}_c}(b_1,b_2)$ is the action at the classical path and $\tilde{A}_N$ is the Hamilton-Jacobi matrix operator. 
\end{theorem}
\begin{proof}
All matrices used below are explicitly defined in Appendix A. Throughout the proof we assume all $\vec{p}_i$ and $\vec{q}_i$ satisfy equations \eqref{eq:hamD1}-\eqref{eq:DBC2}. To begin we directly compute the derivative of the action at the classical path with respect to $\vec{b_1}$,
\begin{align*}
    \frac{\partial \tilde{S}_{d,\tilde{\gamma}_c}(b_1,b_2)}{\partial \vec{b}_1} &= \sum_{i=1}^{N-1} \derp{i}{1} \left(\q_{i+1} - \q_i\right) + \sum_{i=1}^{N-1} \p_i\left(\derq{i+1}{1} - \derq{i}{1}\right)-\sum_{i=1}^{N-1} \derh{\p}{i}(\p_i,\q_i)\derp{i}{1} \\
    & \quad\quad + \sum_{i=1}^{N-1} \derh{\q}{i}(\p_1,\q_i) \derq{i}{1} - \derf{2}{\q}{N}(\q_N,\vec{b}_2)\derq{N}{1} + \derf{1}{\vec{b}}{1} + \derf{1}{\q}{1}(\q_1,\vec{b}_1)\derq{1}{1}
\end{align*}
Once we realize that the above derivative is taken at the classical path, many terms cancel. The first sum cancels with the third sum by equation \eqref{eq:hamD1} and if we rearrange the second sum to be, 
\begin{equation*}
    \sum_{i=1}^{N-1} \p_i\left(\derq{i+1}{1} - \derq{i}{1}\right) = -\derq{1}{1}\p_1 - \sum_{i=2}^{N-1} \derq{i}{1} \left(\p_i-\p_{i-1}\right)+\derq{N}{1}\p_N
\end{equation*}
we see the above cancels out many of the other terms by \eqref{eq:hamD1}, \eqref{eq:DBC1}, and \eqref{eq:DBC2} and so we obtain, 
\begin{equation}
    \frac{\partial \tilde{S}_{d,\tilde{\gamma}_c}(b_1,b_2)}{\partial \vec{b}_1}
    = \derf{1}{\vec{b}}{1} \label{eq:dSdb1}
\end{equation}
Next taking the derivative with respect to $\vec{b}_2$ yields, 
\begin{equation}
    \frac{\partial^2 \tilde{S}_{d,\tilde{\gamma}_c}(b_1,b_2)}{\partial \vec{b}_1 \partial \vec{b}_2} = \left( \derderf{1}{\q}{1}{\vec{b}}{1} \right)^T \left( \derq{1}{2} \right) \label{eq:dSdb1db2}
\end{equation}
Note that the right hand side of equation \eqref{eq:dSdb1} is truthfully, 
\begin{equation*}
    \derf{1}{\vec{b}}{1} = \derf{1}{\vec{b}}{1}(\q,\vec{b}_1) \bigg\rvert_{\q = \q_1}
\end{equation*}
and thus it does not concern the dependence of $\q_1$ on $\vec{b}_1$. This will be the case whenever we write derivatives of $f_1$ or $f_2$ with respect to $\vec{b}_1$ or $\vec{b}_2$.\\

\noindent We would now like to replace $\partial \vec{q}_1/\partial \vec{b}_2$ in equation \eqref{eq:dSdb1db2}. To do so we will take the derivatives of equations \eqref{eq:hamD1}-\eqref{eq:DBC2} with respect to the Lagrangian parameter $\vec{b}_2$,
\begin{align}
    &\derq{i+1}{2} - \derq{i}{2} - \derderh{\p}{i}{\p}{i}\derp{i}{2} + \derderh{\p}{i}{\q}{i}\derq{i}{2} = 0 \label{eq:hamD1der}\\
    &\derp{i}{2} - \derp{i-1}{2} + \derderh{\q}{i}{\q}{i}\derq{i}{2} + \derderh{\p}{i}{\q}{i}\derp{i}{2} = 0 \label{eq:hamD2der}\\
    &\derderf{1}{\q}{1}{\q}{1} \derq{1}{2} = \derp{1}{2} + \derderh{\q}{1}{\q}{1}\derq{1}{2} + \derderh{\q}{1}{\p}{1}\derp{1}{2} \label{eq:DBC1der}\\
    &\derderf{2}{\q}{N}{\vec{b}}{2} + \derderf{2}{\q}{N}{\q}{N}\derq{N}{2} = \derp{N-1}{2} \label{eq:DBC2der}
\end{align}
First it will be useful to write equations \eqref{eq:hamD1der} and \eqref{eq:hamD2der} as the following recursive system of equations, 
\begin{equation}
    \begin{pmatrix}\derq{i+1}{2} \\[0.5em] \derq{i}{2} \end{pmatrix} = U_i \begin{pmatrix}\derq{i}{2} \\[0.5em] \derq{i-1}{2} \end{pmatrix} \label{eq:Usystem}
\end{equation}
where $U_i$ is the $2n \times 2n$ block matrix,
\begin{equation*}
    U_i = \begin{pmatrix}
    \alpha_i & \beta_i \\
    \mathbb{I} & 0 
    \end{pmatrix}
\end{equation*}
and the matrices $\alpha_i$ and $\beta_i$ are given by the equations, 
\begin{align*}
    \alpha_i &= \left(\mathbb{I} + \derderh{\p}{i}{\q}{i}\right) - \derderh{\p}{i}{\p}{i}\left(\mathbb{I}+\derderh{\p}{i}{\q}{i}\right)^{-1}\derderh{\q}{i}{\q}{i} + \derderh{\p}{i}{\p}{i}\left(\mathbb{I}+\derderh{\p}{i}{\q}{i}\right)^{-1}\left(\derderh{\p}{i-1}{\p}{i-1}\right)^{-1}\\[0.5em] 
    \beta_i &= -\derderh{\p}{i}{\p}{i}\left(\mathbb{I}+\derderh{\p}{i}{\q}{i}\right)^{-1}\left(\derderh{\p}{i-1}{\p}{i-1}\right)^{-1}\left(\mathbb{I}+\derderh{\p}{i-1}{\q}{i-1}\right)
\end{align*}
Note that there are no derivatives of $\p_i$ with respect to $\vec{b}_2$ in equation \eqref{eq:Usystem}, as we can substitute equation \eqref{eq:hamD1der} in equation \eqref{eq:hamD2der} to eliminate it. Next we define the vector $W_1$, the initial vector of the recursive system, by, 
\begin{equation*}
    W_1 \derq{1}{2} = \begin{pmatrix}
    \derq{2}{2} \\[0.5em] \derq{1}{2} \end{pmatrix}
\end{equation*}
and so explicitly we have,
\begin{equation*}
    W_1 = \begin{pmatrix} \left(\mathbb{I}+\derderh{\p}{1}{\q}{1}\right) + \derderh{\p}{1}{\p}{1}\left(\mathbb{I}+\derderh{\p}{1}{\q}{1}\right)^{-1}\left(\derderf{1}{\q}{1}{\q}{1}-\derderh{\q}{1}{\q}{1}\right) \\[0.5em] \mathbb{I}
    \end{pmatrix}
\end{equation*}
Combining $W_1$ with the system in equation \eqref{eq:Usystem} we have the useful relation,
\begin{equation*}
    \begin{pmatrix} \derq{N}{2} \\[0.5em] \derq{N-1}{2} \end{pmatrix} = U_{N-1} \cdots U_2W_1\derq{1}{2}
\end{equation*}
Next we rewrite equation \eqref{eq:DBC2der} by rearranging the terms and writing $\partial \p_{N-1} / \partial \vec{b}_2$ in terms of $\partial \q_{N-1} / \partial \vec{b}_2$ and $\partial \q_{N-2} / \partial \vec{b}_2$,
\begin{equation*}
    \derderf{2}{\q}{N}{\vec{b}}{2} = W_2^T \begin{pmatrix} \derq{N}{2} \\[0.5em] \derq{N-1}{2} \end{pmatrix}
\end{equation*}
Putting this all together we get the following convenient way of expressing equation \eqref{eq:DBC2der},
\begin{equation}
    \derderf{2}{\q}{N}{\vec{b}}{2} = \left(W_2^TU_{N-1}\cdots U_2W_1\right)\derq{1}{2} \label{eq:df2withU}
\end{equation}
Observe in the one dimensional case ($n=1$), the matrix product in \eqref{eq:df2withU} is a scalar. Generally, this matrix product gives an $n \times n$ matrix. Plugging equation \eqref{eq:df2withU} back into equation \eqref{eq:dSdb1db2} yields,
\begin{equation*}
    \frac{\partial^2 \tilde{S}_{d,\tilde{\gamma}_c}(b_1,b_2)}{\partial \vec{b}_1 \partial \vec{b}_2} = \left( \derderf{1}{\q}{1}{\vec{b}}{1} \right)^T \left(W_2^TU_{N-1}\cdots U_2W_1\right)^{-1}\left(\derderf{2}{\q}{N}{\vec{b}}{2}\right)
\end{equation*}
and taking the determinant gives, 
\begin{equation}
    \det\left(\frac{\partial^2 \tilde{S}_{d,\tilde{\gamma}_c}(b_1,b_2)}{\partial \vec{b}_1 \partial \vec{b}_2}\right) = \frac{\det\left( \derderf{1}{\p}{1}{\vec{b}}{1} \right)\det\left(\derderf{2}{\p}{N}{\vec{b}}{2}\right)}{\det\left(W_2^TU_{N-1}\cdots U_2W_1\right)} \label{eq:detSwithU}
\end{equation}
Now let's write the denominator of \eqref{eq:detSwithU} in terms of the determinant of the Hamilton-Jacobi matrix operator, $\tilde{A}_N$. To do so we will need the following technical lemma,

\begin{lemma}
For the $(2Nn-n) \times (2Nn-n)$ Hamilton-Jacobi matrix $\tilde{A}_N$,
\begin{equation}
    \det{\tilde{A}_N} = (-1)^{Nn}\left[\prod_{i=1}^{N-1}\det\left(-\derderh{\p}{i}{\p}{i}\right)\right]\det\left(V_2^TT_{N-1}\cdots T_2V_1\right) \det\left(B_{N-1}\cdots B_1\right)  \label{eq:lemmaformula}
\end{equation}
Where we define the block matrices,
\begin{align*}
    T_i &= \begin{pmatrix} -B_i^{-1}E_i & -B_i^{-1}C_{i-1} \\ \mathbb{I} & 0 \end{pmatrix}\\[0.5em]
    V_1 &= \begin{pmatrix} -B_1^{-1}E_1 \\ \mathbb{I} \end{pmatrix} \\[0.5em]
    V_2 &= \begin{pmatrix} -E_N \\ -C_{N-1} \end{pmatrix}
\end{align*}
and the $m \times m$ matrices, 
\begin{align*}
    E_i &= \begin{cases} \derderf{1}{\q}{1}{\q}{1} - \derderh{\q}{1}{\q}{1} + \left(\mathbb{I}+\derderh{\p}{1}{\q}{1}\right)\left(\derderh{\q}{1}{\q}{1}\right)^{-1}\left(\mathbb{I}+\derderh{\p}{1}{\q}{1}\right) & \text{ if } i = 1 \\[0.5em] 
    -\derderh{\q}{i}{\q}{i} + \left(\derderh{\p}{i-1}{\p}{i-1}\right)^{-1}+\left(\mathbb{I}+\derderh{\p}{i}{\q}{i}\right)\left(\derderh{\p}{i}{\p}{i}\right)^{-1}\left(\mathbb{I}+\derderh{\p}{i}{\q}{i}\right) & \text{ if } 2 \leq i \leq N-1 \\[0.5em]
    -\derderf{2}{\q}{N}{\q}{N}+\left(\derderh{\p}{N-1}{\p}{N-1}\right)^{-1} & \text{ if } i = N
    \end{cases}\\[0.5em]
    B_i &= \left(\mathbb{I}+\derderh{\p}{i}{\q}{i}\right)\left(\derderh{\p}{i}{\p}{i}\right)^{-1} \\[0.5em]
    C_i &= \left(\derderh{\p}{i}{\p}{i}\right)^{-1}\left(\mathbb{I} + \derderh{\p}{i}{\q}{i}\right)
\end{align*}
\end{lemma}

\noindent The above lemma is proved in Appendix B. An easy computation reveals the relationships,
\begin{align*}
    V_i &= \begin{pmatrix}-\mathbb{I} & 0 \\ 0 & \mathbb{I} \end{pmatrix}W_i \\[0.5em]
    T_i &= \begin{pmatrix}-\mathbb{I} & 0 \\ 0 & \mathbb{I} \end{pmatrix}U_i\begin{pmatrix}\mathbb{I} & 0 \\ 0 & -\mathbb{I} \end{pmatrix}
\end{align*}
and so we can rewrite equation \eqref{eq:lemmaformula} in terms of the $W$ and $U$ matrices,
\begin{equation}
    \det{\tilde{A}_N} = \left[\prod_{i=1}^{N-1} \det\left(-\derderh{\p}{i}{\p}{i}\right)\det B_i \right] \det\left(W_2^TU_{N-1}\cdots U_2W_1\right) \label{eq:dettildeAWU}
\end{equation}
Using the definition of the $B_i$ matrices and plugging the above into equation (24) we obtain,
\begin{equation*}
    \det\left(\frac{\partial^2 \tilde{S}_{\tilde{\gamma}_c}(b_1,b_2)}{\partial \vec{b}_1 \partial \vec{b}_2}\right) = \left[\prod_{i=1}^{N-1} \det\left(-\mathbb{I} - \derderh{\p}{i}{\q}{i}\right)\right] \frac{\det\left( \derderf{1}{\q}{1}{\vec{b}}{1} \right)\det\left(\derderf{2}{\q}{N}{\vec{b}}{2}\right)}{\det\tilde{A}_N}
\end{equation*}
which is precisely the statement from Theorem II.1.
\end{proof}

\noindent We will be particularly interested in the case where $H(p_i,q_i) = \frac{1}{2m}p_i^2 + V(q_i)$ where the statement from theorem II.1 simplifies to,
\begin{equation*}
    \det\left(\frac{\partial^2 \tilde{S}_{\tilde{\gamma}_c}(b_1,b_2)}{\partial \vec{b}_1 \partial \vec{b}_2}\right) = (-1)^{n(N-1)} \frac{\det\left( \derderf{1}{\q}{1}{\vec{b}}{1} \right)\det\left(\derderf{2}{\q}{N}{\vec{b}}{2}\right)}{\det\tilde{A}_N} 
\end{equation*}
Morevoer, we will now assume $N$ is odd, so the above formula becomes
\begin{equation}
    \det\left(\frac{\partial^2 \tilde{S}_{\tilde{\gamma}_c}(b_1,b_2)}{\partial \vec{b}_1 \partial \vec{b}_2}\right) = \frac{\det\left( \derderf{1}{\q}{1}{\vec{b}}{1} \right)\det\left(\derderf{2}{\q}{N}{\vec{b}}{2}\right)}{\det\tilde{A}_N} \label{eq:GYspecialH}
\end{equation}

\subsection{A Discrete Version of the Operator $A$} \label{section2.3}
Now we will consider the operator $A$ with boundary conditions given by equation \eqref{eq:LBCmixed}. We will define a discretized version of $A$ and compare the determinant of this (finite) operator to the determinant of $\tilde{A}_N$. After proving convergence of these operators in the continuum limit, we will be able to compare the limits of the discrete determinants to the regularized determinants.\\

\noindent For the one dimensional case, we define the discretized version of the operator $A$ as,
\begin{equation}
    (A_N)_{jk} = \begin{cases}-1 & \text{if } j=k+1 \text{ or } k=j+1\\\frac{a_1}{m}+1-\frac{1}{m}V''_j & \text{if } j=k=1\\2-\frac{1}{m}V''_j & \text{if } j = k \text{ and } 2 \leq j \leq N-1\\-\frac{a_2}{m}+1 & \text{if } i=k=1 \\ 0 & \text{otherwise}\end{cases} \label{eq:defAN}
\end{equation}
where $V''_j = V''(q_j)$ and 
\begin{equation*}
    a_1 = \frac{\partial^2 f_1}{\partial q^2}(q,b_1) \quad \quad a_2= \frac{\partial^2 f_2}{\partial q'^2}(q',b_2)
\end{equation*}
Note that this is the usual definition in the case where $\epsilon = 1$. In section 3 we will expand this definition for arbitrary $\epsilon$ in order to consider the convergence of the operator (and its determinant).
\begin{theorem}
Consider the discrete operators $A_N$ and $\tilde{A}_N$, along with the corresponding Hamiltonian is $H(p_i,q_i) = \frac{1}{2m}p_i^2 + V(q_i)$. Their determinants are related by the following formula for all $N \geq 2$.
\begin{equation}
    \det \tilde{A}_N = (-1)^{N-1} m \det A_N \label{eq:relateANtildeAN}
\end{equation}
\end{theorem}
\begin{proof}
The result follows immediately from the fact that, $\det \tilde{A}_N = \det D_1 \det (D_4 - D_3D_1^{-1}D_2)$ and the observation that, $D_4-D_3D_1^{-1}D_2 = m \cdot A_N$ for all $N$ and all twice differentiable function $V(q_i)$.
\end{proof}
An immediate consequence of Theorem II.3 is the following formula,
\begin{corollary}
For the discrete operator $A_N$ with associated Hamiltonian $H(p_i,q_i) = \frac{1}{2m}p_i^2 + V(q_i)$ and mixed boundary conditions from \eqref{eq:LBCmixed}, the following discrete generalized Gelfand-Yaglom formula holds
\begin{equation*}
    \det(A_N) = \frac{1}{m}\frac{\frac{\partial^2 f_1}{\partial b_1 \partial q_1}\frac{\partial^2 f_2}{\partial b_2 \partial q_N}}{\frac{\partial^2 S_{\gamma_c}(b_1,b_2)}{\partial b_1 \partial b_2}}
\end{equation*}
\end{corollary}
\section{Asymptotics and a Lattice Regularization} \label{section3}

In this section we will show that the discrete operators $\tilde{A}_N$ and $A_N$ converge to their continuous counterparts in the continuum limit. Moreover, we will show that we can make sense of the determinants of the $A_N$ and $\tilde{A}_N$ in this limit. This will lead us to define a lattice regularization in regards to the determinants of these operators.\\

\noindent As in section \ref{section2.3}, we will be considering the one-dimensional case where $N$ is odd and $H(p_i,q_i) = \frac{1}{2m}p_i^2+V(q_i)$. We will also employ the following notation as short hand,
\begin{equation}
    a_1 = \derderf{1}{q}{1}{q}{1}(q_1,b_1) \quad \quad a_2 = \derderf{2}{q}{N}{q}{N}(q_N,b_2) \label{eq:a1a2def}
\end{equation}

\subsection{Convergence of $\tilde{A}_N$} \label{section3.1}
Here we consider the operator $\tilde{A}$ given by equation \eqref{eq:optildeA}. We denote the associated twice differentiable, continuous Hamiltonian by $\mathcal{H}\big(p(t),q(t)\big)$. The operator $\tilde{A}$ acts on the domain,
\begin{equation*}
    D\big(\tilde{A}\big) = \left\{ \begin{pmatrix}x_1(t) \\ x_2(t)\end{pmatrix} \; \Big| \; x_1,x_2\in C^1\big([0,T]\big) \text{ and }  x_1(0) = a_1 x_2(0), \, x_1(T) = a_2 x_2(T)\right\}
\end{equation*}
where the last two conditions are just the boundary conditions stated in \eqref{eq:LBCvec}.\\ 

\noindent The associated discrete operator, $\tilde{A}_N$ arises from the discrete Hamilitonian $H(p_i,q_i) = \epsilon \cdot \mathcal{H}\big(p(t_i),q(t_i)\big)$. Recall that the parameter $\epsilon = \frac{T}{N-1}$ splits the interval $[0,T]$ into $N$ equally spaced time points. The domain of $\tilde{A}_N$ is, 
\begin{equation*}
    D(\tilde{A}_N) = \left\{ \begin{pmatrix}x_1(t_1) \\ \vdots \\ x_1(t_{N-1}) \\ x_2(t_1) \\ \vdots \\ x_2(t_N) \end{pmatrix} \;:\; \begin{pmatrix}x_1(t) \\ x_2(t)\end{pmatrix} \in D\big(\tilde{A}\big) \right\}
\end{equation*}
\begin{theorem}
The discrete operator $\tilde{A}_N$ converges weakly to the operator $\tilde{A}$ as $N \to \infty$ for any twice differentiable Hamiltonian $\mathcal{H}\big(p(t),q(t)\big)$. 
\end{theorem}
\begin{proof}
Let's first define the vectors $X,Y \in D\big(\tilde{A}\big)$ as
\begin{equation*}
    X = \begin{pmatrix} x_1(t) \\ x_2(t) \end{pmatrix}, \quad Y = \begin{pmatrix} y_1(t) \\ y_2(t) \end{pmatrix}
\end{equation*}
and the corresponding vectors $X_N, \, Y_N \in D(\tilde{A}_N)$ as,
\begin{equation*}
    X_N = \begin{pmatrix} x_1(t_1) \\ \vdots \\ x_1(t_{N-1}) \\ x_2(t_1) \\ \vdots \\ x_2(t_N)\end{pmatrix}, \quad Y_N = \begin{pmatrix} y_1(t_1) \\ \vdots \\ y_1(t_{N-1}) \\ y_2(t_1) \\ \vdots \\ y_2(t_N) \end{pmatrix}
\end{equation*}
To show weak convergence, we will show that
\begin{equation}
    \lim_{N \to \infty} Y_N^T\mathcal{D}_NX_N = \int_0^T Y^T \tilde{A} X \, dt \label{eq:weakconv}
\end{equation}
We compute that,
\begin{align*}
    Y^T_N\tilde{A}_NX_N &= -\sum_{i=1}^{N-1} \epsilon y_1(t_i)\frac{\partial^2 \mathcal{H}}{\partial p^2}x_1(t_i) + \sum_{i=1}^{N-1} \epsilon y_1(t_i)\left[\left(\frac{x_2(t_{i+1})-x_2(t_i)}{\epsilon}\right)-\frac{\partial^2\mathcal{H}}{\partial p \partial q}x_2(t_i)\right] \\
    & \quad\quad - \sum_{i=1}^{N-1} \epsilon y_2(t_i) \frac{\partial^2 \mathcal{H}}{\partial q^2} x_2(t_i) - \sum_{i=1}^{N-1} \epsilon y_2(t_i) \frac{\partial^2 \mathcal{H}}{\partial p \partial q}x_1(t_i) - \sum_{i=1}^{N-2}\epsilon y_2(t_i)\left(\frac{x_1(t_{i+1})-x_1(t_i)}{\epsilon}\right) \\
    & \quad\quad\quad\quad -\big(x_1(t_1) - a_1x_2(t_1)\big) + \big(x_1(t_{N-1}) - a_2x_2(t_N)\big) 
\end{align*}
Now taking the limit gives, 
\begin{align*}
    \lim_{N\to\infty} Y^T_N\tilde{A}_NX_N &= -\int_0^T y_1(t)\frac{\partial^2 \mathcal{H}}{\partial p^2}x_1(t) dt + \int_0^T y_1(t)\left[x_2'(t)-\frac{\partial^2 \mathcal{H}}{\partial p \partial q}x_2(t)\right]dt \\
    &\quad \quad - \int_0^T y_2(t)\left[x_1'(t) + \frac{\partial^2 \mathcal{H}}{\partial p \partial q}\right]dt - \int_0^T y_2(t) \frac{\partial^2 \mathcal{H}}{\partial p^2} x_2(t) dt \\
    &\quad\quad\quad\quad - \big(x_1(0) - a_1x_2(0)\big) + \big(x1(T) - a_2x_2(T)\big) \\
    \lim_{N\to\infty} Y^T_N\tilde{A}_NX_N &= \int_0^T Y^T\tilde{A}X dt - \big(x_1(0) - a_1x_2(0)\big) + \big(x_1(T) - a_2x_2(T)\big) 
\end{align*}
The boundary terms are zero for all $X \in D\big(\tilde{A}\big)$ and so the above statement is exactly equation \eqref{eq:weakconv}.
\end{proof}

\noindent Now let us restrict to the case where $\mathcal{H}\big(p(t),q(t)\big) = \frac{1}{2m}p(t)^2 + V\big(q(t)\big)$. In this case, taking the limit of equation \eqref{eq:GYspecialH} (under the convention that $N$ is odd) gives,
\begin{equation}
    \lim_{N\to\infty} \det\tilde{A}_N = \frac{ \derderf{1}{q}{}{b}{1} \derderf{2}{q'}{}{b}{2}}{\frac{\partial^2 \tilde{S}_{\tilde{\gamma}_c}(b_1,b_2)}{\partial b_1 \partial b_2}} \label{eq:latticeregtildeA}
\end{equation}
The right hand side of the above equation is well-defined and finite, therefore the limit on the left hand side is also well-defined and finite. We will use this limit later in section \ref{section3.3} to define lattice-regularization. Note that the convergence of this limit is no longer clear in the case of a Hamiltonian with mixed terms and needs further understanding.

\subsection{Convergence of $A_N$} \label{section3.2}
We now return to the operator $A$ from equation \eqref{eq:opA} and its finite counterpart $A_N$. 
\begin{theorem}
The operator $A_N$ weakly converges to the operator $A$.
\end{theorem}
\begin{proof}
We first must define the domains of the operators $A$ and $A_N$. The operator $A$ has the domain, 
\begin{equation*}
    D(A) = \left\{ y(t) \in C^2([0,T]) \; | \; y'(0) = a_1 \cdot y(0), \; y'(T) = a_2 \cdot y(T) \right\}
\end{equation*}
where the mixed boundary conditions match the boundary conditions on $\tilde{A}$ given by equation \eqref{eq:LBCvec}. The domain of the operator $A_N$ is,
\begin{equation*}
    D(A_N) = \left\{\begin{pmatrix}y(t_1) \\ \vdots \\ y(t_N) \end{pmatrix} \; \Bigg| \; y(t) \in D(A) \right\}
\end{equation*}
Previously when defining $A_N$ we used the convention $\epsilon = 1$, so we first need to reinsert epsilons into $A_N$ where appropriate. For the case of $N = 4$ and $A = -\frac{d^2}{dt^2} - \frac{1}{m}V''(q_c(t))$ the operator $A_N$ is,
\begin{equation*}
    A_4 = 
    \begin{pmatrix}
    \frac{a_1}{m} + \frac{1}{\epsilon} - \frac{\epsilon}{m}V''\big(q_c(t_1)\big) & -\frac{1}{\epsilon} & 0 & 0 \\
    -\frac{1}{\epsilon} & \frac{2}{\epsilon} -\frac{\epsilon}{m}V''\big(q(t_2)\big) & -\frac{1}{\epsilon} & 0 \\
    0 & -\frac{1}{\epsilon} & 2 -\frac{\epsilon}{m}V''\big(q(t_2)\big) & -\frac{1}{\epsilon} \\
    0 & 0 & -\frac{1}{\epsilon} & -\frac{a_2}{m} + \frac{1}{\epsilon} 
    \end{pmatrix}
\end{equation*}
The above is easily generalized for arbitrary $N$. Let $x(t),y(t) \in D(A)$ and let $X_N,Y_N \in D(A_N)$ be their corresponding discrete versions. We will show that, 
\begin{equation}
    \lim_{N\to\infty} Y_N^TA_NX_N = \int_0^T y(t)Ax(t) dt \label{eq:weakconv2}
\end{equation}
First we compute,
\begin{align*}
    Y_NA_NX_N &= y(t_1)\left(-\frac{x(t_2) - x(t_1)}{\epsilon} + \frac{a_1}{m}x(t_1)\right) - \sum_{i=1}^{N-1} y(t_i)\frac{\epsilon}{m}V''\big(q_c(t_i)\big)x(t_i) \\
    &\quad\quad - \sum_{i=2}^{N-1} \epsilon y(t_i)\left(\frac{x(t_{i+1}) - 2x(t_i) + x(t_{i-1})}{\epsilon^2}\right) + y(t_N)\left(\frac{x(t_N) - x(t_{N-1})}{\epsilon} - \frac{a_2}{m}x(t_N)\right)
\end{align*}
Taking the limit yields, 
\begin{align*}
    \lim_{N\to\infty} Y_NA_NX_N &= y(0)\left(x'(0) + \frac{a_1}{m}x(0)\right) - \int_0^T y(t)\frac{1}{m}V''\big(q_c(t)\big)x(t)dt \\
    &\quad\quad\quad\quad - \int_0^T y(t)x''(t) dt + y(T)\left(x'(T) - \frac{a_2}{m}x(T)\right) \\
    \lim_{N\to\infty} &= \int_0^T y(t)Ax(t)dt + y(0)\left(x'(0) + \frac{a_1}{m}x(0)\right) + y(T)\left(x'(T) - \frac{a_2}{m}x(T)\right)
\end{align*}
The boundary terms are zero for all $x(t) \in D(A)$ and so the above statement is exactly equation \eqref{eq:weakconv2}.
\end{proof}

\noindent Again we will restrict to the case of $\mathcal{H}\big(p(t),q(t)\big) = \frac{1}{2m}p(t)^2 + V\big(q(t)\big)$. After generalizing for arbitrary $\epsilon$, equation \eqref{eq:relateANtildeAN} becomes
\begin{equation*}
    \det \tilde{A}_N = m \epsilon^{N-1}  \det A_N
\end{equation*}
where we must now be cognisant of the epsilons in $A_N$ and $\tilde{A}_N$. Plugging this into \eqref{eq:GYspecialH}, in order to get something convergent we must take a regularized determinant where we throw out the factor of $\epsilon^{-N+1}$,
\begin{equation*}
    \lim_{N\to\infty} \det{}' A_N = \frac{ \derderf{1}{q}{}{b}{1} \derderf{2}{q'}{}{b}{2}}{m\frac{\partial^2 \tilde{S}_{\tilde{\gamma}_c}(b_1,b_2)}{\partial b_1 \partial b_2}}
\end{equation*}
where the apostrophe indicates that we have removed the epsilons. Again, the right hand side above is well-defined and finite.\\

\noindent It should be noted that for arbitrary $\epsilon$, the determinant of $\tilde{A}_N$ converges plainly, however the determinant of $A_N$ does not. In the latter case we need to remove the divergence. This might motivate the Hamilton-Jacobi operator being a more natural choice over Laplacian-type operators. 

\subsection{Defining a Lattice Regularization} \label{section3.3}
As show in sections \ref{section3.1} and \ref{section3.2}, one can make meaning out of the limits $\lim_{N\to\infty} \det \tilde{A}_N$ and $\lim_{N\to\infty} \det A_N$ in the case where $H(p_i,q_i) = \frac{p_i^2}{2m} + V(q_i)$. The following definition is a natural consequence,
\begin{definition}
We define the lattice regularized determinants of $A$ and $\tilde{A}$ by,
\begin{equation}
    \det{}_{\text{reg}}(A) = \lim_{N\to\infty} \det{}'(A_N)
\end{equation}
\begin{equation}
    \det{}_{\text{reg}}\big(\tilde{A}\big) = \lim_{N\to\infty} \det\big(\tilde{A}_N\big)
\end{equation}
\end{definition}
\noindent Tautologically, we have the identity
\begin{equation}
    \det{}_\text{reg} \big(\tilde{A}\big) = m\det{}_\text{reg}\big(A\big) \label{eq:detAtildeArelate}
\end{equation}
The above definitions accompanied with equation \eqref{eq:latticeregtildeA} give use a generalized GY formula for the lattice regularized determinant of the operator $A$,
\begin{equation}
    \det{}_\text{reg} A = \frac{ \derderf{1}{q}{}{b}{1} \derderf{2}{q'}{}{b}{2}}{m\frac{\partial^2 \tilde{S}_{\tilde{\gamma}_c}(b_1,b_2)}{\partial b_1 \partial b_2}} \label{eq:GYlatticeregA}
\end{equation}

\section{A Generalized Gelfand-Yaglom Formula for the Zeta Regularization} \label{section4}

\noindent In this section we will first derive a Gelfand Yaglom formula for the $\zeta$-regularized determinant of the second order operator $L = -\frac{d^2}{dt^2} + u(t)$ equipped with mixed boundary conditions. While this formula is not new (in fact, it was first derived more generally by Burghelea, Friedlander, and Kappeler\cite{Kep}) we will specifically relate it to the operator $A$ with relevant boundary conditions. Moreover, we will compare the results to the formula in equation \eqref{eq:GYlatticeregA}.

\subsection{Derivation of a generalized GY formula for the configuration space} \label{section4.1}
Let $u(t) \in C^1\left([0,T],\mathbb{R}\right)$.  We will consider the differential operator, 
\begin{equation}
    L = - \frac{d^2}{dt^2} + u(t) \label{eq:Lop}
\end{equation}
on the interval $t \in [0,T]$ with the domain,
\begin{equation}
    D(L) = \left\{y(t) \in W^{2,2}(0,T) \; : \; \frac{dy(0)}{dt} = \frac{a_1}{m} y(0), \; \frac{dy(T)}{dt} = \frac{a_2}{m} y(T) \right\} \label{eq:domainL}
\end{equation}
where $W^{2,2}(0,T)$ denotes the Sobelov space and $a_1$, $a_2$ and $m$ are nonzero constants named suggestively. We will also need to consider the second order differential equation,
\begin{equation}
    -\ddot{y} + u(t)y = \lambda y \label{eq:diffeqy}
\end{equation}
with parameter $\lambda$ and where a dot denotes the derivative with respect to $t$. Let $y_1(t,\lambda)$ and $y_2(t,\lambda)$ denote two solutions of \eqref{eq:diffeqy} with the following boundary conditions,
\begin{equation}
    y_1(0,\lambda) = 1, \; \dot{y}_1(0,\lambda) = \frac{a_1}{m}
\end{equation}
\begin{equation}
    y_2(T,\lambda) = 1, \; \dot{y}_2(T,\lambda) = \frac{a_2}{m}
\end{equation} 
We are now able to state the following result, which is a specialization of a theorem first proved by Burghelea, Friedlander, and Kappeler\cite{Kep},
\begin{theorem}
Let $y_1(t) = y_1(t,0)$ be the solution given above. Then,
\begin{equation}
    \det{}_\zeta \, L = 2 \left(\dot{y}_1(T) - \frac{a_2}{m} y_1(T)\right)
\end{equation}
where $L$ is the differential operator defined by equations \eqref{eq:Lop} and \eqref{eq:domainL}.
\end{theorem}
\begin{proof}
Let's start by taking a closer look at the differential operator. The operator $L$ is a regular Sturm-Liouville operator and thus has a discrete spectrum with simple eigenvalues, $\lambda_1 < \lambda_2 < \cdots < \lambda_n < \cdots$, accumulating to $\infty$. Moreover, for large $n$,
\begin{equation*}
    \lambda_n = \frac{\pi^2n^2}{T^2} + O(1)
\end{equation*}
The details of this can be found in \textit{Sturm-Liousville and Dirac Operators}\cite{Sturm} by Levitan and Sargsyan among other texts. It then follows that the resolvent of $L$, $R_{\lambda} = (L - \lambda I)^{-1}$, is a trace class operator. So we can write the useful relation,
\begin{equation}
    \frac{d}{d\lambda}\log \det{}_\zeta (L - \lambda I) = -\text{Tr } R_\lambda \label{eq:zetadettrace}
\end{equation}
where any zero eigenvalues are removed. Using variation of parameter on the inhomogeneous equation,
\begin{equation*}
    -\ddot{y}+u(t)y = \lambda y + f(x), \quad \lambda \neq \lambda_n
\end{equation*}
we get the solution, 
\begin{equation*}
    y(x) = \int_0^T R_\lambda(x,\xi)f(\xi)d\xi
\end{equation*}
where 
\begin{equation}
    R_\lambda(x,\xi) = \begin{cases}
    \frac{y_1(x,\lambda)y_2(\xi,\lambda)}{W(y_1,y_2)} & \text{ if } x\leq \xi \\
    \frac{y_1(\xi,\lambda)y_2(x,\lambda)}{W(y_1,y_2)} & \text{ if } x\geq \xi \\
    \end{cases}
\end{equation}
is the resolvent of $L$. In the above, $W(y_1,y_2)$ denotes the Wronskian of the two solutions. We manipulate the right hand side of equation \eqref{eq:zetadettrace} as follows,
\begin{align*}
    -\text{Tr } R_\lambda &= -\int_0^T R_\lambda(x,x) dx \\
        &= -\frac{1}{W(y_1,y_2)}\int_0^T y_1(x,\lambda)y_2(x,\lambda)dx \\
        &= \frac{1}{W(y_1,y_2)}\left[W\left(\frac{dy_1}{d\lambda},y_2\right)\right]_0^T \\
        &= \frac{d}{d\lambda} \log \left[\frac{a_2}{m} y_1(T,\lambda) - \dot{y}_1(T,\lambda)\right]
\end{align*}
Plugging the above back into \eqref{eq:zetadettrace} gives, 
\begin{equation}
    \det{}_\zeta(L - \lambda I) = C \cdot \left[\frac{a_2}{m}y_1(T,\lambda) - \dot{y}_1(T,\lambda)\right] \label{eq:whatsC}
\end{equation}
where $C$ is some constant. To compute $C$, we will let $\lambda = -\mu$ and consider the asymptotics of both sides of the equation as $\mu \to \infty$. To start let's compute the asymptotics of $\det(L+\mu I)$. We write the $\zeta$-function of $L+\mu I$ using the contour integral method described by Kirsten\cite{FuncDet},
\begin{equation}
    \zeta_{L+\mu I}(s) = \frac{1}{2\pi i} \int_{\gamma} dx x^{-s} \frac{d}{dx} \log \omega (x-\mu)
\end{equation}
where $\gamma$ is the curve encircling all the eigenvalues of $L + \mu I$ and $\omega(x-\mu)$ is a smooth function of $x$ with zero at the eigenvalues of the operator $L + \mu I$. Let $\sqrt{x} = \sigma + ri$, then
\begin{equation}
    \omega(x) = - \sqrt{x} \sin \left(T\sqrt{x}\right) + O\left(e^{|r|T}\right)
\end{equation}
The full computation of these asymptotics can be found in work by Fulton and Pruess\cite{FulPru}. Next we deform the contour and we rewrite the integral as, 
\begin{equation}
    \zeta_{L+\mu I}(s) = \frac{\sin (\pi s)}{\pi} \int_0^\infty dx x^{-s} \frac{d}{dx} \log \omega (-x-\mu)
\end{equation}
The above integral converges near 0 for $s=0$, however the integral does not converge near infinity for $s=0$. To analytically continue the function we write,
\begin{equation*}
    \zeta_{L+\mu I}(s) = \zeta_1(s) + \zeta_2(s) + \zeta_3(s)
\end{equation*}
where
\begin{align*}
    \zeta_1(s) &= \frac{\sin (\pi s)}{\pi} \int_0^1 dx x^{-s} \frac{d}{dx} \log \omega (-x-\mu) \\
    \zeta_2(s) &= \frac{\sin (\pi s)}{\pi} \int_1^\infty dx x^{-s} \frac{d}{dx} \log \left(\omega (-x-\mu)\frac{2}{\sqrt{x}}e^{-T\sqrt{x}}\right) \\
    \zeta_3(s) &= \frac{\sin (\pi s)}{\pi} \int_1^\infty dx x^{-s} \frac{d}{dx} \log \left(\frac{1}{2}\sqrt{x}e^{T\sqrt{x}}\right)
\end{align*}
The first two integrals converge for $s=0$ and we can easily analytically continue the third using the method described by Kirsten\cite{FuncDet}. Using the above we compute,
\begin{equation*}
    \zeta'_{L+\mu I}(0) = -\log 2w(-\mu)
\end{equation*}
and so the determinant is,
\begin{equation*}
    \det{}_\zeta(L+\mu I) = 2w(-\mu)
\end{equation*}
Asymptotically we can write, 
\begin{equation}
    \det{}_\zeta(L+\mu I) = 2\sqrt{\mu} \sinh\left(T\sqrt{\mu}\right) + O\left(e^{T\sqrt{\mu}}\right) \label{eq:asymDetL}
\end{equation}
Now we will consider the right hand side of equation \eqref{eq:whatsC}. Again, let $\sqrt{x}=\sigma+ri$. We will also let $k = \int_0^Tu(t)dt$. The function $y_1(T,x)$ has the following asymptotic expansions (again proven by Fulton and Pruess\cite{FulPru}),
\begin{equation*}
    y_1(T,x) = \cos(T\sqrt{x}) + \left(\frac{a_1}{m\sqrt{x}}+\frac{k}{2\sqrt{x}}\right)\sin(T\sqrt{x}) + O\left(\frac{1}{|x|}e^{T|r|}\right)
\end{equation*}
And so equation \eqref{eq:whatsC} becomes,
\begin{equation}
   \det{}_\zeta(L+\mu I) = C\cdot \left[-\sqrt{\mu}\sinh(T\sqrt{\mu})+O(e^{T\sqrt{\mu}})\right] \label{eq:asypy1}
\end{equation}
Comparing equations \eqref{eq:asymDetL} and \eqref{eq:asypy1}, we see that $C = -2$, so equation \eqref{eq:whatsC} becomes
\begin{equation}
    \det{}_\zeta(L - \lambda I) =2 \left( \dot{y}_1(T,\lambda) - \frac{a_2}{m}y_1(T,\lambda)\right)
\end{equation}
In particular, if we consider the case of $\lambda = 0$, we obtain the result from theorem IV.1.
\end{proof}

\noindent Note that in the case where $a_1 = a_2 = 0$ we recover the case of Neumann boundary conditions. The case of Dirichlet boundary conditions cannot be extracted from the above theorem, however the result is well known\cite{QMmath}. Let us now relate the above formula to the quantum system described in section \ref{section1.1} with Lagrangian boundary conditions.  First let,
\begin{equation*}
    u(t) = -\frac{1}{m}V''\left(q_c(t)\right) 
\end{equation*}
where $q_c(t)$ is the classical path. Note the classical path has the initial conditions, 
\begin{equation}
    q_c(0) = q, \quad \dot{q}_c(0) = \frac{1}{m} \frac{\partial f_1}{\partial q} \label{eq:classpathIC}
\end{equation}
which leads us to the following lemma,
\begin{lemma}
The function  $y(t) = \frac{\partial q_c(t)}{\partial q}$ with boundary conditions,
\begin{equation*}
    y(0) = 1 \quad \dot{y}(0) = \frac{a_1}{m}
\end{equation*}
satisfies the differential equation,
\begin{equation*}
    m\ddot{y}(t) = -V''\left(q_c(t)\right)y(t)
\end{equation*}
\end{lemma}
\noindent The above lemma is a simple exercises in derivatives. The following corollary is an immediate result of Theorem IV.1 and Lemma IV.2, 
\begin{corollary}
For the operator $A$ with domain given by \eqref{eq:domainL} we have the generalized Gelfand-Yaglom formula,
\begin{equation*} 
    \det{}_\zeta \, A = 2 \left(\frac{\partial \dot{q}_c(T)}{\partial q} - \frac{a_2}{m} \frac{\partial q_c(T)}{\partial q} \right)
\end{equation*}
where $q_c(t)$ is the classical path satisfying equations \eqref{eq:ELeq} and \eqref{eq:classpathIC}.
\end{corollary}

\subsection{A generalized GY formula for the phase space and $A$} \label{section4.2}
In this section, we will reformulate Corollary IV.2.1 to be in terms of derivatives of the action functional from \eqref{eq:actionH2}. We claim,
\begin{theorem}
For the action given by equation \eqref{eq:actionH2} with Hamiltonian $H(p,q) = \frac{p^2}{2m}+V(q)$, the following generalized Gelfand-Yaglom formula holds,
\begin{equation}
    \frac{\partial^2 \tilde{S}_{\tilde{\gamma}_c}(b_1,b_2)}{\partial b_1 \partial b_2} = 2\frac{\derderf{1}{q}{}{b}{1} \derderf{2}{q'}{}{b}{2}}{m \det{}_\zeta \, A} 
\end{equation}
\end{theorem}
\begin{proof}
Let's start by taking derivatives of the action at the critical value,
\begin{equation}
    \frac{\partial^2 \tilde{S}_{\tilde{\gamma}_c} (b_1,b_2)}{\partial b_1 \partial b_2} = \derderf{1}{b}{1}{q}{} \frac{\partial q}{\partial b_2} \label{eq:derderS}
\end{equation}
Recall the second boundary condition from equation \eqref{eq:LBC}. Taking the derivative with respect to $b_2$ gives,
\begin{equation*}
    \derderf{2}{q'}{}{b}{2} + a_2 \cdot \frac{\partial q'}{\partial b_2} = \frac{\partial p'}{\partial b_2}
\end{equation*}
The above uses the notation $q(T) = q'$, $p(T) = p'$, and the shorthand given in equation \eqref{eq:a1a2def}. Let us rewrite the above, using relation $p(t) = m\dot{q}(t)$.
\begin{equation*}
    \derderf{2}{q'}{}{b}{2} = - a_2 \cdot \frac{\partial q'}{\partial q}\frac{\partial q}{\partial b_2} + m \cdot \frac{\partial \dot{q}'}{\partial q}\frac{\partial q}{\partial b_2}
\end{equation*}
Now let's use Corollary IV.2.1 to replace the right hand side of the above, 
\begin{equation*}
    \derderf{2}{q'}{}{b}{2} = \left(\frac{m}{2} \det{}_\zeta A\right) \cdot \frac{\partial q}{\partial b_2}
\end{equation*}
Solving for $\partial q/\partial b_2$ and plugging the results into equation \eqref{eq:derderS} yields the statement in the theorem.
\end{proof}
\noindent The following corollary is an immediate consequence of Theorem IV.3,
\begin{corollary}
The lattice-regularize determinant and $\zeta$-regularized determinant of $A$ relate in the following manner,
\begin{equation*}
    \det{}_{\text{reg}}A = \frac{1}{2}\det{}_\zeta A
\end{equation*}
\end{corollary}

\section{Concluding remarks}
\noindent The lattice-regularization, compared to zeta regularization, gives an alternative, possibly more natural, method of regularization for the Hamilton-Jacobi operator. It follows from equations (24) and (43) that the two methods are closely related. In cases of Hamiltonians with mixed derivatives, the lattice-regularization presents a potentially easier method of computing the regularized determinant (where we do not have a typical Gelfand-Yaglom formula). However, the convergence of \eqref{eq:DGYgen} needs to be better understood in this case. \\

\noindent As seen in Theorem IV.1, one immediate shortcoming of this work is it does not obviously relate to the case of Dirichlet boundary conditions. To obtain similar results that relate to the Dirichlet case, one must switch the roles of $p$ and $q$ in the Lagrangian boundary conditions and re-derive most of the formulas. While most of the details will follow immediately from the work here, it would take a concerted effort. \\

\noindent In the future, we would like to extend these results to the quantum field theory setting. Again, we hope that convergence of a descrete formula will give an alternative method for computing the regularized determinant for field theories. One could also consider generalizing these results for more general boundary conditions such as those considered by Burghelea, Friedlander, and Kappeler\cite{Kep}.

\section*{Acknowledgements} 
I would like to thank my advisor, Nicolai Reshetikhin, for all the guidance and Leon Takhtajan for providing thoughtful comments. This work was partly supported by the NSF FRG Collaborative Research Grant DMS-1664387 and thanks to the hospitality of ETH-ITS.
\section*{Data Availability Statement}
Data sharing is not applicable to this article as no new data were created or analyzed in this study.

\appendix 
\section{Description of matrices}

Let us explicitly describe the matrices used throughout the proof of Theorem 2.1. We denote the vectors, $\p_i = (p_i^1, p_i^2, \dots, p_i^n)$ and $\q_i = (q_i^1, q_i^2, \dots, q_i^n)$. We define,
\begin{align*}
    \left(\derderh{\p}{k}{\p}{k}\right)_{ij} = \derderh{p^i}{k}{p^j}{k} &\quad\quad \left(\derderh{\p}{k}{\q}{k}\right)_{ij} = \derderh{p^i}{k}{q^j}{k} \\
    \left(\derderh{\q}{k}{\q}{k}\right)_{ij} = \derderh{q^i}{k}{q^j}{k} &\quad\quad \left(\derderf{1}{\q}{1}{\q}{1}\right)_{ij} = \derderf{1}{q^i}{1}{q^j}{1} \\
    \left(\derderf{2}{\q}{N}{\q}{N}\right)_{ij} = \derderf{2}{q^i}{N}{q^i}{N} &\quad\quad
    \left(\derderf{1}{\q}{1}{\vec{b}}{1}\right)_{ij} = \derderf{1}{q^i}{1}{b^j}{1} \\
    \left(\derderf{1}{\q}{N}{\vec{b}}{2}\right)_{ij} = \derderf{2}{q^i}{N}{b^j}{2} &\quad\quad \left(\derq{k}{2}\right)_{ij} = \frac{\partial q^i_k}{\partial b^j_2}
\end{align*}
where all of the above matrices are $n \times n$.

\section{Proof of Lemma 2.2}
Here we will assume the matrices $\partial^2 H / \partial \p_i{}^2$ are invertible for all $i=1,\dots, n$. Thus we can write the determinant of $\tilde{A}_N$ as, 
\begin{equation*}
    \det(\tilde{A}_N) = \det(D_1)\det(D_4-D_3D_1^{-1}D_2)
\end{equation*}
where the matrices $D_i$ for $i = 1, 2, 3, 4$. are described by equation \eqref{eq:blocktildeAN}. The matrix $D_4 - D_3D_1^{-1}D_2$ is a block tridiagonal matrix and so we may write the resulting determinant as described in \cite{0712.0681},
\begin{equation*}
    \det(D_4-D_3D_1^{-1}D_2) = (-1)^{Nm} \det(T_{11}^{(0)})\det(B_1 \cdots B_{N-1})
\end{equation*}
Where the matrices $B_i$ are defined in section 2.2 and the matrix $T_{11}^{(0)}$ is given by, 
\begin{equation*}
    T_{11}^{(0)} = V_2^TT_{N-1}\cdots T_2V_1
\end{equation*}
From the above statement, Lemma 2.2 is clear.

\bibliography{GYpaper}

\end{document}